\documentclass[review]{elsarticle}

\usepackage{lineno,hyperref}
\modulolinenumbers[5]

\journal{Journal of Discrete Algorithms}

\usepackage[utf8]{inputenc}
\usepackage{xcolor,colortbl}
\usepackage{xspace}
\usepackage{multirow}
\usepackage{arydshln}
\usepackage{hhline}
\usepackage{url}
\usepackage{listings}
\usepackage{wasysym}
\usepackage{amssymb}

\newtheorem{thm}{Theorem}
\newtheorem{lem}{Lemma}
\newdefinition{rmk}{Remark}
\newproof{proof}{Proof}

\newcommand{\qedwhite}{\hfill \ensuremath{\Box}}

\newcommand{\sigmas}{\ensuremath{\Sigma^+}\xspace}
\newcommand{\sigmad}{\ensuremath{\Sigma^{\$}}\xspace}
\newcommand{\lcp}{\ensuremath{lcp}\xspace}

\newcommand{\SA}{\ensuremath{\mathsf{SA}}\xspace}
\newcommand{\LCP}{\ensuremath{\mathsf{LCP}}\xspace}
\newcommand{\PLCP}{\ensuremath{\mathsf{PLCP}}\xspace}
\newcommand{\LCPs}{\ensuremath{\LCP\kern-.10em(T_s)}\xspace}
\newcommand{\LCPsn}{\ensuremath{\LCP\kern-.10em{(T_{s+1})}}\xspace}

\newcommand{\LCPt}{\ensuremath{\mathsf{LCP}_2}\xspace}
\newcommand{\LCPts}{\ensuremath{\LCPt\kern-.10em(T_s)}\xspace}
\newcommand{\LCPtsn}{\ensuremath{\LCPt\kern-.10em{(T_{s+1})}}\xspace}

\newcommand{\RMQ}{\ensuremath{rmq}\xspace}
\newcommand{\BWT}{\ensuremath{\mathsf{BWT}}\xspace}

\newcommand{\etal}{{\it et al.}\xspace}
\newcommand{\ie}{{\it i.e.}\xspace}

\begin{document}

\begin{frontmatter}

\title{Burrows-Wheeler transform and {LCP} array construction in constant space\tnoteref{mytitlenote}}

\tnotetext[mytitlenote]{A preliminary version of this work appeared in IWOCA 2015~\cite{Louza2015}.
}




\author[unicamp]{Felipe A. Louza\corref{mycorrespondingauthors}}
\cortext[mycorrespondingauthors]{Corresponding author}
\ead{louza@ic.unicamp.br}
\author[diego-portales]{Travis Gagie}
\ead{travis.gagie@mail.udp.cl}
\author[unicamp]{Guilherme P. Telles}
\ead{gpt@ic.unicamp.br}

\address[unicamp]{Institute of Computing, University of Campinas, Campinas, SP, Brazil}
\address[diego-portales]{School of Informatics and Telecommunications, Diego Portales University, Santiago, Chile\\
CeBiB (Center of Biotechnology and Bioengineering), Santiago, Chile}

\sloppy
\begin{abstract}
In this article we extend the elegant in-place Burrows-Wheeler transform (BWT)
algorithm proposed by Crochemore et al.~(Crochemore et al., 2015).
Our extension is twofold: 
we first show how to compute simultaneously the longest common prefix (LCP)
array as well as the BWT, using constant additional space; 
we then show how to build the LCP array directly in compressed representation  
using Elias coding, still using constant additional space and with no asymptotic slowdown.
Furthermore, we provide a time/space tradeoff for our algorithm when additional memory is
allowed.
Our algorithm runs in quadratic time, as does
Crochemore et al.'s, and is supported by interesting properties of the BWT and
of the LCP array, contributing to our understanding of the time/space tradeoff
curve for building indexing structures. 
\end{abstract}


\end{frontmatter}


\section{Introduction}

There have been many
articles~\cite{Okanohara2009,Tischler2014,Kasai2001,Karkkainen2009,Beller2013b}
about building the Burrows-Wheeler transform (\BWT)~\cite{Burrows1994} and  the
longest common prefix (\LCP) array.  
For example, Belazzougui~\cite{Belazzougui2014b} showed how we can compute the \BWT and
the (permuted) \LCP array of a string $T$ of length $n$ over an alphabet of size
$\sigma$ in linear time and $O (n\log\sigma)$ bits of space (see also~\cite{Belazzougui2016,Munro2017}).
Navarro and Nekrich~\cite{Navarro2014b} and Policriti, Gigante and
Prezza~\cite{Policriti2015} showed how to build the \BWT in compressed space
and, respectively, $O (n \log n / \log \log n)$ worst-case time and
average-case time proportional to the length of the compressed representation of $T$.

The most space-efficient \BWT construction algorithm
currently known, however, is due to Crochemore \etal~\cite{Crochemore2015}: it
builds the \BWT in place --- \ie, replacing the input string with the \BWT ---
in $O (n^2)$ time for unbounded alphabets
using only a constant number of $\Omega(\log n)$ bit words of additional
memory (\ie, four integer variables and one character variable). 
Unlike most \BWT-construction algorithms, this one is symmetric to the \BWT inversion.
Its simplicity and elegance make it very attractive from a theoretical point of
view and it is interesting as one extreme of the time/space tradeoff curve for
building BWTs.  
Because a quadratic time bound is impractical, however, Crochemore \etal
showed how their algorithm can be speeded up at the cost of using more space.

Closely related to the \BWT, the suffix array
(\SA)~\cite{Manber1993,Gonnet1992} may be constructed by many algorithms in 
linear time 
(see~\cite{Puglisi2007,Dhaliwal2012,Karkkainen2016b} for reviews).
Franceschini and Muthukrishnan~\cite{Franceschini2007} presented a 
suffix array construction algorithm that runs in $O(n\log n)$ time using
constant additional space.
The \LCP array can be computed in linear time together with \SA during
the suffix sorting~\cite{Fischer2011,Louza2017} or independently given
$T$ and \SA as input~\cite{Kasai2001,Manzini2004b,Karkkainen2009} or
given the \BWT~\cite{Gog2011b,Beller2013b}.
Table~\ref{t:related} summarizes the most closely related algorithms'
bounds.\footnote{Although the authors did not mention it, it seems likely
Navarro and Nekrich's and Policriti \etal's algorithms can also be made to
reuse the space occupied by the text for the \BWT.  With that modification,
their \(n H_k (T) + o (n \log \sigma)\) space bounds in Table~\ref{t:related}
can be made \(o (n \log \sigma)\).}

\begin{table}[t]
\centering
\caption{
Summary of related works and their theoretical bounds.
The last column shows the additional space used on top of what is needed to
store the input and the output. 
Belazzougui's algorithm~\cite{Belazzougui2014b} was randomized but has been
made deterministic~\cite{Belazzougui2016,Munro2017}.
Navarro and Nekrich's algorithm~\cite{Navarro2014b} uses \(n H_0 (T) + o (n
\log \sigma)\) bits on top of the text, where \(H_k (T) \leq
\lg \sigma\) is the $k$th-order empirical entropy of $T$.  
Policriti \etal's algorithm~\cite{Policriti2015} uses \(n H_k (T) + n + O
(\sigma \log n) + o(n \log \sigma)\) bits on top of the text and runs in \(O (n
(H_k (T) + 1))\) time in the average case.  
For discussion of empirical entropy, see, e.g.,~\cite{Ohlebusch2013,Makinen2015,Navarro2016}.  
For simplicity, in this table we assume \(\sigma \in \omega(1) \cap o(n/\log
n)\).
} \label{t:related}
\resizebox{1\textwidth}{!}{
\begin{tabular}{l|c|c|c|l|l}
\hline
                        & \BWT          & \LCP & \SA & time   & additional space          \\
\hline
Belazzougui~\cite{Belazzougui2014b}  			&  $\checkmark$  & $\checkmark$	&  		& $O(n)$ & $O(n \log \sigma)$ bits \\
Navarro and Nekrich~\cite{Navarro2014b} 		&  $\checkmark$  & 	     	&  		& $O (n \log n / \log \log n)$       &      \(n H_k (S) + o (n \log \sigma)\) bits              \\
Policriti~\etal~\cite{Policriti2015}   			&  $\checkmark$	 & 	     	&  		& $O (n (H_k (S) + 1) (\log n / \log \log n)^2)$      	 & \(n H_k (S) + o(n \log \sigma)\) bits                   \\
Crochemore \etal~\cite{Crochemore2015}  		&  $\checkmark$  & 		&  		& $O(n^2)$       &  $O(1)$                  \\
Franceschini and Muthukrishnan~\cite{Franceschini2007}  &     		 & 		&  $\checkmark$	& $O(n\log n)$	       &  $O(1)$                  \\
Fischer~\cite{Fischer2011}  &     		 & $\checkmark$ 		&  $\checkmark$	& $O(n)$	       &  $O(n \log n)$  bits                  \\
Louza \etal~\cite{Louza2017}  &     		 & $\checkmark$ 		&  $\checkmark$	& $O(n\sigma)$	       &  $O(\sigma \log n)$ bits                   \\
 {\bf Our algorithm}  					&  $\checkmark$  & $\checkmark$	& 		& $O(n^2)$ & $O(1)$               \\
\hline
\end{tabular}
}
\end{table}

In this article we show how Crochemore \etal's algorithm can be extended to
compute also the longest common prefix (\LCP) array of a string $T$ of length $n$.  
Specifically, we show how, given $\BWT(T[i + 1,n-1])$ and $\LCP(T[i +1,n-1])$ 
and $T[i]$, we can compute $\BWT(T[i,n-1])$ and $\LCP(T[i,n-1])$ using $O(n - i)$
time and constant extra space on top of what is needed to store
$\BWT(T[i,n-1])$ and $\LCP(T[i,n-1])$.  
Our construction algorithm has many of the nice properties of Crochemore
\etal's original: it is conceptually simple and in-place, it allows practical
time-space tradeoffs
\footnote{We are aware that the LCP array and the BWT array can be
computed with similar worst-case bounds by using a combination of Franceschini
and Muthukrishnan's algorithm to build the \SA, then computing the \LCP naively
in $O(n^2)$ time overwriting the \SA, and finally using Crochemore \etal's
algorithm to compute the \BWT overwriting the text.  We still think our
algorithm is interesting, however, because of its simplicity --- the C
implementation fits in a single page --- and its offer of encoding and
tradeoffs.}, and we can compute some compressed encodings of
\(\LCP (T)\) directly.  This is particularly interesting because in
practice the \LCP\ array can be compressed by nearly a logarithmic
factor.
Computing the \BWT and \LCP together in small space is interesting,
for example, when building compressed suffix trees 
(see, e.g.~\cite{Sadakane2007,Ohlebusch2010b,Fischer2011b,Abeliuk2013,Gog2013a,Navarro2014a}), which are
space-efficient versions of the classic linear-space suffix tree
\cite{Weiner1973} that is often based on the \BWT and \LCP.

There exist external memory algorithms that compute the
\BWT~\cite{Ferragina2012,Bauer2013}
and the \LCP
array~\cite{Karkkainen2016,Karkkainen2016c,Karkkainen2016d,Bingmann2016}.
In particular, Bauer~\etal~\cite{Bauer2012a} and Cox~\etal~\cite{Cox2016}
showed how to construct the \BWT and the \LCP array simultaneously for string collections.
They compute the \LCP values and process the \BWT in a order similar to the one
we use for the algorithm in this article, but their solution uses auxiliary
memory and partitions the output into buckets to address external-memory access
issues.
Tischler~\cite{Tischler2016} introduced an external-memory algorithm that
computes the Elias~$\gamma$-coded~\cite{Elias1975} permuted \LCP given the \BWT
and the sampled inverse
suffix array as input.
For further discussion, we refer the reader to recent books by
Ohlebusch~\cite{Ohlebusch2013}, M\"akinen \etal~\cite{Makinen2015} and
Navarro~\cite{Navarro2016}.

The rest of the article is organized as follows.
In Section~\ref{s:background} we introduce concepts and notations.
In Section~\ref{s:crochemore} we review the in-place \BWT algorithm by
Crochemore \etal.
In Section~\ref{s:algorithm} we present our algorithm and in
Section~\ref{s:compressed} we show how the \LCP can be constructed in
compressed representation.
In Section~\ref{s:tradeoff} we provide a tradeoff between time and space
for our algorithm when additional memory is allowed.
In Section~\ref{s:conclusion} we conclude the article and we leave an
open question.

\section{Background}
\label{s:background}

Let $\Sigma$ be an ordered alphabet of $\sigma$ symbols. We denote the set
of every nonempty string of symbols in $\Sigma$ by $\sigmas$. 
We use the symbol $<$ for the lexicographic order relation between strings.
Let $\$$ be
a symbol not in $\Sigma$ that precedes every symbol in $\Sigma$. We define
$\sigmad = \{T \$ \mid T \in \sigmas\}$. 

The $i$-th symbol in a string $T$ will be denoted by $T[i]$. Let $T=T[0]
T[1]\ldots T[n-1]$ be a string of length $|T|=n$. A substring of $T$ will be
denoted by $T[i,j] = T[i] \ldots T[j]$, $0\leq i\leq j < n$. A prefix
of $T$ is a substring of the form $T[0,k]$ and a suffix is a substring of
the form $T[k,n-1]$, $0\leq k< n$.
The suffix $T[k,n-1]$ will be denoted by $T_k$.

\subsection*{Suffix array, LCP array and the BWT}

A suffix array for a string provides the lexicographic order for all its
suffixes. Formally, a suffix array \SA for a string $T\in\sigmad$ of size
$n$ is an array of integers $\SA=[i_0, i_1, \ldots, i_{n-1}]$ such that
$T_{i_0} < T_{i_1} < \ldots < T_{i_{n-1}}$~\cite{Manber1993,Gonnet1992}.

Let $\lcp(S,T)$ be the length of the longest common prefix of two strings
$S$ and $T$ in $\sigmad$. The \LCP array for $T$ stores the value of
$\lcp$ for suffixes pointed by consecutive positions of a suffix array.
We define $\LCP[0] = 0$ and $\LCP[i] = \lcp(T_{\SA[i]},T_{\SA[i-1]})$ for
$1\leq i < n$.

The \BWT of a string $T$ can be constructed by listing all the $n$ circular
shifts of $T$, lexicographically sorting them, aligning the shifts
columnwise and taking the last column~\cite{Burrows1994}. The \BWT is
reversible and tends to group identical symbols in runs. It may also be defined in
terms of the suffix array, to which it is closely related. Let the \BWT of
a string $T$ be denoted simply by \BWT. We define $\BWT[i]=T[\SA[i]-1]$ if
$\SA[i]\neq 0$ or $\BWT[i]=\$$ otherwise.

The first column of the conceptual matrix of the \BWT will be referred to as $F$,
and the last column will be referred to as $L$.
The LF-mapping property of the \BWT states that the
$i^{th}$ occurrence of a symbol $\alpha \in \Sigma$ in $L$
corresponds to the $i^{th}$ occurrence of $\alpha$ in $F$.

Some other relations between the \SA and the \BWT are the following.
It is easy to see that $L[i]= \BWT[i]$ and $F[i]=T[\SA[i]]$.
Moreover, if the first symbol of $T_{\SA[i]}$, $T[\SA[i]]=\alpha$,
is the $k^{th}$ occurrence of $\alpha$ in $F$, then
$j$ is the position of $T_{\SA[i]+1}$ in \SA (\ie $j$ is the rank of
$T_{\SA[i]+1}$) such that $L[j]$ corresponds to the $k^{th}$ occurrence of
$\alpha$ in $L$.

As an example, Figure~\ref{fig:esa} shows the circular shifts, the sorted
circular shifts, the \SA, the \LCP, the \BWT and the sorted suffixes for
$T=\mbox{BANANA\$}$.

\begin{figure}[ht]
\resizebox{\textwidth}{!} {
\centering
\begin{tabular}{c|c|c|c|c|c|l}
  \multicolumn{1}{c}{} & \multicolumn{1}{c}{~~circular shifts~~} & \multicolumn{1}{c}{~~sorted circular shifts~~} & \multicolumn{3}{c}{} & \multicolumn{1}{c}{~sorted suffixes~}\\ 
  \hline
  $i$ &                           & {$F$~~~~~~~~~~~$L$} &  ~\SA~ & ~\LCP~ & ~\BWT~ & ~$T_{\SA[i]}$    \\
  \hline
 0  & {BANANA\$}  & {\$BANANA}  & 6 & 0 & A  & ~\$        \\
 1  & {\$BANANA}  & {A\$BANAN}  & 5 & 0 & N  & ~A\$       \\
 2  & {A\$BANAN}  & {ANA\$BAN}  & 3 & 1 & N  & ~ANA\$     \\
 3  & {NA\$BANA}  & {ANANA\$B}  & 1 & 3 & B  & ~ANANA\$   \\
 4  & {ANA\$BAN}  & {BANANA\$}  & 0 & 0 & \$ & ~BANANA\$  \\
 5  & {NANA\$BA}  & {NA\$BANA}  & 4 & 0 & A  & ~NA\$      \\
 6  & {ANANA\$B}  & {NANA\$BA}  & 2 & 2 & A  & ~NANA\$    \\
  \hline
\end{tabular}
}
\caption{\SA, \LCP and \BWT for $T=\mbox{BANANA\$}$.}
\label{fig:esa}
\end{figure}

The range minimum query (\RMQ) with respect to the \LCP is the smallest \lcp
value in an interval of a suffix array. We define $\RMQ(i,j) =
\min_{i < k \leq j} \{\LCP[k]\}$, for $0 \le i <j < n$. Given a string $T$ of
length $n$ and its \LCP array, it is easy to see that
$\lcp(T_{\SA[i]},T_{\SA[j]}) = \RMQ(i,j)$.

\subsection*{Elias coding}

The Elias $\gamma$-code of a positive number $\ell \ge 1$ is composed of
the unary-code of $\lfloor \log_2 \ell \rfloor + 1$ (a sequence of $\lfloor 
\log_2 \ell \rfloor$ 0-bits ended by one 1-bit), followed by the binary
code of $\ell$ without the most significant bit \cite{Witten1999}. The
$\gamma$-code encodes $\ell$ in $2 \lfloor \log_2 \ell\rfloor + 1$ bits.
For instance, $\gamma(4) = 00100$, since the unary code for $\lfloor\log_2
4\rfloor + 1 = 3$ is $001$ and $4$ in binary is $100$.

The Elias $\delta$-code of $\ell$ is composed of the $\gamma$-code of
$1+\lfloor \log_2 \ell \rfloor$, followed by the binary code of $\ell$
without the most significant bit. The $\delta$-coding represents $\ell$
using $2 \lfloor \log_2( \lfloor \log_2 \ell \rfloor +1)\rfloor + 1 +
\lfloor \log_2 \ell \rfloor$ bits, which is asymptotically optimal
\cite{Elias1975}. For instance, $\delta(9) = 00100001$, since
$\gamma(\lfloor\log_2 9\rfloor + 1) = 00100$ and $9$ in binary is $1001$.

Decoding a $\gamma$-encoded number $\ell_{\gamma}$ requires finding the
leftmost 1-bit in the unary code of $\lfloor \log_2 \ell \rfloor + 1$, and
interpreting the next $\ell-1$ bits as a binary code. Decoding a
$\delta$-encoded number $\ell_{\delta}$ requires decoding a
$\gamma$-code and then reading the proper number of following bits as a
binary code. Both decodings may be performed in constant time in a CPU
having instructions for counting the number of leading zeros and shifting a
word by an arbitrary number of bits.

\section{In-place BWT}
\label{s:crochemore}

The algorithm by Crochemore \etal~\cite{Crochemore2015}
overwrites the input string $T$ with the \BWT as it proceeds by induction
on the suffix length.

Let $\BWT(T_s)$ be the \BWT of the suffix $T_s$, stored in $T[s,n-1]$. The
base cases are the two rightmost suffixes, for which $\BWT(T_{n-2})=T_{n-2}$
and $\BWT(T_{n-1})=T_{n-1}$. For the inductive step, the authors have shown
that the position of $\$$ in $\BWT(T_{s+1})$ is related to the rank of
$T_{s+1}$ among the suffixes $T_{s+1}$,\ldots,$T_{n-1}$ (local rank), thus
allowing for the construction of
$\BWT(T_{s})$ even after $T[s+1,n-1]$ has been overwritten
with $\BWT(T_{s+1})$. The algorithm comprises four steps.

\begin{description}
\item[1]
Find the position $p$ of \$ in $T[s+1,n-1]$. Evaluating $p-s$ gives the
local rank of $T_{s+1}$ that originally was starting at position $s+1$.

\item[2]
Find the local rank $r$ of the suffix $T_s$ using just symbol
$c = T[s]$.
To this end, sum the number of symbols in $T[s+1,n-1]$ that are strictly smaller
than $c$ with the number of occurrences of $c$ in $T[s+1,p]$ and with $s$. 

\item[3]
Store $c$ into $T[p]$, replacing \$.

\item[4]
Shift $T[s+1,r]$ one position to the left. Write \$ in $T[r]$.
\end{description}

The algorithm runs in $O(n^2)$ time using constant space memory. Furthermore,
the algorithm is also in-place since it uses $O(1)$ additional memory and
overwrites the input text with the output \BWT.

\section{LCP array in constant space}\label{s:algorithm}

Our algorithm computes both the \BWT and the \LCP array by induction on the
length of the suffix. The \BWT construction is the same as proposed by
Crochemore \etal~\cite{Crochemore2015}.
Let us first introduce an overview of our algorithm.

At a glance, the \LCP evaluation works as follows. Suppose that
$\BWT(T_{s+1})$ and the \LCP array for the suffixes
$\{T_{s+1},\ldots,T_{n-1}\}$, denoted by $\LCPsn$, have already been built.
Adding the suffix $T_s$ to the solution requires evaluating exactly two
values of \lcp, involving the two suffixes that will be adjacent to $T_s$.

The first \lcp value involves $T_s$ and the largest suffix $T_a$ in
$\{T_{s+1},\ldots,T_{n-1}\}$ that is smaller than $T_s$.
Fortunately, $\BWT(T_{s+1})$ and $\LCPsn$ are sufficient 
to compute such value.
Recall that if the first symbol of $T_a$ is not equal to the first
symbol of $T_s$ then $lcp(T_a, T_s) = 0$.
Otherwise $lcp(T_a, T_s) = lcp(T_{a+1}, T_{s+1})+1$ and the
$\RMQ$ may be used, since both $T_{a+1}$ and $T_{s+1}$ are already in $\BWT(T_{s+1})$.
We know that the position of $T_{s+1}$ is $p$ from Step 1 of the in-place \BWT
in Section \ref{s:crochemore}.
Then it is enough to find, in $\BWT(T_{s+1})$, the position of $T_{a+1}$, which
stores the symbol corresponding to the first symbol of $T_{a}$.

The second \lcp value involves $T_s$ and the smallest suffix $T_b$ in
$\{T_{s+1},\ldots,T_{n-1}\}$ that is larger than $T_s$. It may be computed in
a similar fashion.

\subsection*{Basic algorithm}

Suppose that $\BWT(T_{s+1})$ and $\LCPsn$ have already been built and are
stored in $T[s+1,n-1]$ and $\LCP[s+1,n-1]$, respectively. Adding $T_s$,
whose rank is $r$, to the solution requires updating $\LCPsn$: by first
shifting $\LCP[s+1, r]$ one position to the left and then computing the new
values of $\LCP[r]$ and $\LCP[r+1]$, which refer to the two suffixes
adjacent to $T_s$ in $\LCPs$.

The value of $\LCP[r]$ is equal to the \lcp of $T_s$ and $T_a$ in
$\BWT(T_{s+1})$. The rank of $T_a$ is $r$ and will be $r-1$ in $\BWT(T_s)$
after shifting. If the first symbol of $T_a$ is equal to $T[s]$ then
$\LCP[r] = \lcp(T_{a+1}, T_{s+1})+1$, otherwise $\LCP[r] = 0$.

We can evaluate $\lcp(T_{a+1}, T_{s+1})$ by the \RMQ
function from the position of $T_{a+1}$ to the position of
$T_{s+1}$.
We know that $p$ is the position of $T_{s+1}$ in $\BWT(T_{s+1})$.
Then we must find the position $p_{a+1}$ of $T_{a+1}$ in
$\BWT(T_{s+1})$.

Note that $T[p_{a+1}]$ corresponds to the first symbol of $T_a$. If
$T[p_{a+1}]~\ne~T[s]$ then $\lcp(T_a, T_s) = 0$, otherwise the value of
$\lcp(T_a, T_s)$ may be evaluated as $\lcp(T_{a+1}, T_{s+1})+1 = \RMQ(p_{a+1},
p)+1$.

The value of $\LCP[r+1]$ may be evaluated in a similar fashion. Let $T_b$
be the suffix with rank $r+1$ in $\BWT(T_{s+1})$ (its rank will still be
$r+1$ in $\BWT(T_s)$). We must find the position $p_{b+1}$ of
$T_{b+1}$ in $\BWT(T_{s+1})$ and then if $T[s] = T[p_{b+1}]$
compute $\LCP[r+1] = \lcp(T_{s}, T_{b}) =
\lcp(T_{s+1}, T_{b+1})+1 = \RMQ(p, p_{b+1})+1$.

The algorithm proceeds by induction on the length of the suffix. It is
easy to see that for the suffixes with length 1 and 2, the values in \LCP
will be always equal to 0. Let the current suffix be $T_s$ ($0 \le s \le
n-3$).
Our algorithm has new Steps 2', 2'' and 4', added just after Steps 2 and 4,
respectively, of the in-place \BWT algorithm as follows:

\begin{description}
\item[2'] Find the position $p_{a+1}$ of the suffix $T_{a+1}$, such that
 suffix $T_{a}$ has rank $r$ in $\BWT(T_{s+1})$, and compute:

  \[
  \ell_a  =
  \left\{
  \begin{array}{ll}
    \RMQ(p_{a+1}, p)+1 & \mbox{~if } T[p_{a+1}] = T[s]\\
    0 &                 \mbox{~otherwise.}
  \end{array}
  \right.
  \]

\item[2''] Find the position $p_{b+1}$ of the suffix $T_{b+1}$, such that
  suffix $T_{b}$ has rank $r+1$ in $\BWT(T_{s+1})$, and compute:
  \[
  \ell_b  =
  \left\{
  \begin{array}{ll}
    \RMQ(p, p_{b+1})+1 & \mbox{~if }  T[s] = T[p_{b+1}]\\
    0                 & \mbox{~otherwise.}
  \end{array}
  \right.
  \]

\item[4'] Shift $\LCP[s+1,r]$ one position to the left, store $\ell_a$ in
  $\LCP[r]$ and if $r+1<n$ then store $\ell_b$ in $\LCP[r+1]$.
\end{description}

\subsection*{Computing $\ell_a$ and $\ell_b$}

To find $p_{a+1}$ and $p_{b+1}$ and to compute $\ell_a$ and $\ell_b$ in
Steps 2' and 2'', we use the following properties.

\begin{lem}
\label{lemma:a}
Let $T_s$ be the suffix to be inserted in $\BWT(T_{s+1})$ at position $r$.
Let ${T_a \in \{T_{s+1},\ldots,T_{n-1}\}}$ be the suffix whose rank is $r$ in
$\BWT(T_{s+1})$, and let $p_{a+1}$ be the position of $T_{a+1}$.
If $p_{a+1} \notin [s+1, p)$ then $T[p_{a+1}] \ne T[s]$.
\end{lem}

\begin{proof}
The local rank of $T_a$ in $\BWT(T_{s+1})$ is $r-s$.
We know that $T[p_{a+1}]$ corresponds to the first symbol of $T_a$, and it
follows from LF-mapping that the local rank of $T[p_{a+1}]$ is $r-s$ in
$\BWT(T_{s+1})$. Then $T[p_{a+1}]$ is smaller than or equal to $T[s]$, since $T_s$
also has local rank $r-s$. If $T[p_{a+1}]$ is smaller than $T[s]$, $p_{a+1}$
must be in $[s+1, n)$. However, if $T[p_{a+1}]=T[s]$ then $p_{a+1}$ must
precede the position where $T[s]$ will be inserted, \ie the position $p$ of
$T_{s+1}$, otherwise the local rank of $T_s$ would be smaller than $r-s$. Then
if $T[p_{a+1}] = T[s]$ it follows that $p_{a+1} \in [s+1, p)$.
\hfill $\qedwhite$
\end{proof}

We can use Lemma~\ref{lemma:a} to verify whether $T[p_{a+1}] = T[s]$ by simply
checking if there is a symbol in $T[s+1, p-1]$ equal to $T[s]$. If no such
symbol is found, $\ell_a = 0$, otherwise we need to compute $\RMQ(p_{a+1}, p)$.
Furthermore, if we have more than one symbol in $T[s+1, p-1]$ equal to $T[s]$,
the symbol whose local rank is $r-s$ will be the last symbol found in $T[s+1, p-1]$,
\ie the largest symbol in $T[s+1, p-1]$ smaller than $T[s]$. Then, to find
such symbol we can simply perform a backward scan in $T$ from $p-1$ to $s+1$
until we find the first occurrence of $T[p_{a+1}]=T[s]$. One can see that we
are able, simultaneously, to compute the minimum function for the \lcp visited values,
obtaining $\RMQ(p_{a+1}, p)$ as soon as we find $T[p_{a+1}]=T[s]$.

\begin{lem}\label{lemma:b}
Let $T_s$ be the suffix to be inserted in $\BWT(T_{s+1})$ at position $r$.
Let ${T_b \in \{T_{s+1},\ldots,T_{n-1}\}}$ be the suffix whose rank is $r+1$ in
$\BWT(T_{s+1})$, and let $p_{b+1}$ be the position of $T_{b+1}$. If $p_{b+1}
\notin (p, n-1]$ then $T[s]\ne T[p_{b+1}]$.
\end{lem}

The proof of Lemma~\ref{lemma:b} is similar to the proof of
Lemma~\ref{lemma:a} and will be omitted.
It is important to remember, though, that
$T_b$ will still
have rank $r+1$ in $\BWT(T_{s})$ (after inserting $T_s$).

The procedure to find $\ell_b$ uses Lemma~\ref{lemma:b} and computes
$\lcp(T_{s+1}, T_{b+1})$ in a similar fashion. It scans $T$
from $p+1$ to $n-1$ until it finds the first occurrence of
$T[p_{b+1}]=T[s]$, computing the minimum function to solve the $\RMQ$ if
such symbol is found.

The C source code presented in Figure \ref{fig:code} implements the algorithm
using eight integer variables apart from the $n \log_2 \sigma$ bits used to store
$T$ and compute the \BWT, and the $n \log_2 n$ bits used to compute the LCP array.
This code is also available at \url{https://github.com/felipelouza/bwt-lcp-in-place}.

\begin{figure}\label{fig:code}
\begin{lstlisting}[]
void compute_bwt_lcp(unsigned char *T, int n, int *LCP){
int i, p, r=1, s, p_a1, p_b1, l_a, l_b;
LCP[n-1] = LCP[n-2] = 0; // base cases

for (s=n-3; s>=0; s--) {

  /*steps 1 and 2*/
  p=r+1;
   for (i=s+1, r=0; T[i]!=END_MARKER; i++)
     if(T[i]<=T[s]) r++;
   for (; i<n; i++)
     if (T[i]<T[s]) r++;

   /*step 2'*/
   p_a1=p+s-1;
   l_a=LCP[p_a1+1];
   while (T[p_a1]!=T[s]) // rmq function
     if (LCP[p_a1--]<l_a)
       l_a=LCP[p_a1+1];
   if (p_a1==s) l_a=0;
   else l_a++;

   /*step 2''*/
   p_b1=p+s+1;
   l_b=LCP[p_b1];
   while (T[p_b1]!=T[s] && p_b1<n) // rmq function
     if (LCP[++p_b1]<l_b)
       l_b=LCP[p_b1];
   if (p_b1==n) l_b=0;
   else l_b++;

   /*steps 3 and 4*/
   T[p+s]=T[s];
   for (i=s; i<s+r; i++) {
     T[i]=T[i+1];
     LCP[i]=LCP[i+1];
   }
   T[s+r]=END_MARKER;

   /*step 4'*/
   LCP[s+r]=l_a;
   if (s+r+1<n)  // If r+1 is not the last position
     LCP[s+r+1]=l_b;
  }
}
\end{lstlisting}
\caption{\BWT and \LCP array construction algorithm}
\end{figure}

\subsection*{Example}

As an example, consider $T= \mbox{BANANA\$}$ and $s=1$.
Figures \ref{fig:step2} and \ref{fig:step4} illustrate Steps 2' and 4', respectively.
The values in red in columns \LCP and \BWT were still not computed.
Suppose that we have
computed $\BWT({T_2})$ and $\LCP({T_2})$. We then have $p=6$ (Step 1) and
the rank $r=4$ (Step 2).

\begin{figure}[ht]
  \centering
  \begin{tabular}{rc|c|c|l}
    & s~ & ~~\LCP~~ & ~\BWT~ & ~sorted suffixes \\
    \hhline{~----}
    & \cellcolor[gray]{0.9}0~ &  \cellcolor[gray]{0.9}\textcolor{red}{-} & \cellcolor[gray]{0.9}\textcolor{red}{B} & ~\cellcolor[gray]{0.9}BANANA\$               \\
   $s \rightarrow $ & \cellcolor[gray]{0.9}1~ &  \cellcolor[gray]{0.9}\textcolor{red}{-} & \cellcolor[gray]{0.9}\textcolor{red}{A} & ~\cellcolor[gray]{0.9}ANANA\$               \\
    \cdashline{2-5}
                          & 2~ & 0   & A   & ~\$ \\
                          & 3~ & 0   & N   & ~A\$ \\
    $r \rightarrow $      & 4~ & 1   & N   & ~ANA\$ \\
    $p_{a+1} \rightarrow $ & 5~ & 0   & A   & ~NA\$ \\
    $p \rightarrow $      & 6~ & 2   & \$  & ~NANA\$ \\
    \cline{2-5}
  \end{tabular}
  \caption{After Step 2'': $T = \mbox{BANANA\$}$ and $s=1$.}
  \label{fig:step2}
\end{figure}

Step 2' finds the first symbol equal to $T[s]$ (A) in $T[s+1,p-1]$ at
position $p_{a+1} = 5$. It represents $T_{a+1}=\mbox{NA}\$$. In this case, the
value of $\ell_a$ is calculated during the scan of $T$ from $p-1=5$ to
$s+1=2$, \ie $\ell_a = \RMQ(p_{a+1},p) = \RMQ(5,6) = 2$. Step 2'' does not
find any symbol equal to $T[s]$ (A) in $T[p+1, n-1]$.
Thus we know that $T[s]\ne T[p_{b+1}]$ and $\ell_b=0$.

\begin{figure}[ht]
  \centering
  \begin{tabular}{rc|c|c|l}
    & s~& ~\LCP~ & ~\BWT~  & ~sorted suffixes \\
    & \cellcolor[gray]{0.9}0~ &  \cellcolor[gray]{0.9}\textcolor{red}{-} & \cellcolor[gray]{0.9}\textcolor{red}{B} & ~\cellcolor[gray]{0.9}BANANA\$               \\
    \cdashline{2-5}
    &       1~ & 0 & A    & ~\$            \\
    &       2~ & 0 & N    & ~A\$   \\
    &       3~ & 1 & N    & ~ANA\$ \\
    ~~~~~ $r \rightarrow $&   \cellcolor[gray]{0.9}4~ & \cellcolor[gray]{0.9}\textcolor{black}{$\ell_a$} = 3& \cellcolor[gray]{0.9}\$  & ~\cellcolor[gray]{0.9}ANANA\$ \\
    &       \cellcolor[gray]{0.9}5~ & \cellcolor[gray]{0.9}\textcolor{black}{$\ell_b$} = 0& \cellcolor[gray]{0.9}A   & ~\cellcolor[gray]{0.9}NA\$                 \\
    &       6~ & 2 & A    & ~NANA\$        \\
    \cline{2-5}
  \end{tabular}
  \caption{After Step 4': $T = \mbox{BANANA\$}$ and $s=1$.}
  \label{fig:step4}
\end{figure}

Step 3 stores $T[s]$ (A) at position $T[p]$, $p=6$. Step 4 shifts
$T[s+1,r]$ one position to the left and inserts \$ at position
$T[r]$, $r=4$. The last step, 4', shifts $\LCP[s+1,r]$ one position to the
left and sets $\LCP[4] = \ell_a = 3$ and $\LCP[4+1]= \ell_b = 0$.

\begin{thm}
\label{thm:bwt_lcp1}
Given a string $T$ of length $n$, we can compute its $\BWT$ in-place and $\LCP$ array
simultaneously in $O(n^2)$ time using $O(1)$ additional space.
\end{thm}

\begin{proof}
The cost added 
by Steps 2' and 2'' were two $O(n)$ time scans over
$T_{s+1}$ to
compute the values of $\ell_a$ and $\ell_b$, whereas Step 4' shifts the \LCP by
the same amount that \BWT is shifted. 
Therefore, the time complexity of our algorithm remains the same 
as the in-place \BWT algorithm, that is, $O(n^2)$. 
As for the space usage, our new algorithm needs only four additional variables
to store positions $p_{a+1}$ and $p_{b+1}$ and the values of $\ell_a$ and
$\ell_b$, thus using constant space only.
\hfill $\qedwhite$
\end{proof}

\section{LCP array in compressed representation}\label{s:compressed}

The \LCP array can be represented using less than $n \log n$ bits. Some
alternatives for encoding the \LCP array store its values in text
order~\cite{Sadakane2002,Fischer2010}, building an array that is known as
permuted \LCP (\PLCP)~\cite{Karkkainen2009}. Some properties of the \PLCP will
allow for encoding the whole array achieving better compression rates. However,
most applications will require the \LCP array itself, and will convert the \PLCP
to the \LCP array~\cite{Gog2013a}.
Other alternatives for encoding the \LCP will preserve its elements'
order~\cite{Abouelhoda2004,Brisaboa2013}. 

Recall that to compute the \BWT and the \LCP array in constant space
only sequential scans are performed. Therefore, the values in the \LCP
array can be easily encoded and decoded during such scans using a
universal code, such as Elias $\delta$-codes [25], with no need to further
adjust the algorithm. Our LCP array representation will encode its
values in the same order, and will be generated directly.

The algorithm will build the \BWT and a compressed \LCP array that will be
called $\LCPt$. $\LCPt$ will be treated as a sequence of bits from this
point on.

The \lcp values will be $\delta$-encoded during the algorithm such that
consecutive intervals $\LCPt[b_{i},e_{i}]$ encode
$\lcp(T_{\SA[i]},T_{\SA[i-1]})+1$.  
We add 1 to guarantee that the values are always positive integers and
can be encoded using $\delta$-codes.
We will assume that decoding subtracts this $1$ added by the encoding
operation.

Suppose that $\BWT(T_{s+1})$ and $\LCPtsn$ have already been built such
that every value in $\LCPtsn$ is $\delta$-encoded and stored in
$\LCP_2[b_{s+1},e_{n-1}]$. Adding $T_s$ to the solution requires
evaluating the values of $\ell_a$ and $\ell_b$ computed in Steps 2' and 2''
and the length of the shift to be performed in $\LCP_2[b_{s+1},e_{r-1}]$ by
Step 4'.

\subsection*{Modified Step 2'}

We know by Lemma~\ref{lemma:a} that if there is no symbol in $T[s+1,p-1]$
equal to $T[s]$, then $\ell_a = 0$, which is encoded as $\delta(0+1) =
1$. Otherwise, if $T[s]$ occurs at position $p_{a+1} \in [s+1, p)$, we
may compute $\RMQ(p_{a+1}, p)$ as the minimum value encoded in
$\LCP_2[b_{p_{a+1}+1}, e_{p}]$. We use two extra variables to store the
positions $b_{s+1}$ and $e_p$ of $\LCP_2$ corresponding to the beginning
of the encoded $\lcp(T_{\SA[s+1]},T_{\SA[s]})+1$ and the ending of the
encoded $\lcp(T_{\SA[p]},T_{\SA[p-1]})+1$. These two variables are
easily updated at each iteration.

As our algorithm performs a backward scan in $T$ to find $T[p_{a+1}]$, we
cannot compute the \RMQ function decoding the \lcp values during this scan.
Therefore, we first search for position $p_{a+1}$ scanning $T$. Then, if
$p_{a+1}$ exists, the first bit of $\LCP_2[p_{a+1}+1]$ is found by
decoding and discarding the first $p_{a+1}-s+1$ encoded values from
$b_{s+1}$. At this point $\RMQ(p_{a+1}, p)$ may be evaluated by finding
the minimum encoded value from $\LCP_2[p_{a+1}+1]$ to $\LCP_2[e_p]$. At
the end, we add $1$ to obtain $\ell_a$.

\subsection*{Modified Step 2''}

The algorithm performs a forward scan in $T$ to find the position $p_{b+1}
\in (p, n-1]$. Analogously to Modified Step 2', we know by Lemma
\ref{lemma:b} that if $T[s]$ does not occur in $T[p+1, n-1]$
then $\ell_b = 0$, which is encoded as $\delta(0+1)$. Otherwise, the \RMQ
over the encoded \lcp values may be computed during this scan. The value
of $\RMQ(p, p_{b+1})$ is computed decoding the values in $\LCP_2$ one by
one, starting at position $e_p+1$ and continuing up to position
$p_{b+1}$ in $T$. At the end, we add 1 to obtain $\ell_b$.

\subsection*{Modified Step 4'}

The amount of shift in the compressed \LCPtsn must account for the sizes of
$\delta(\ell_a+1)$, of $\delta(\ell_b+1)$ and of the encoding of the \lcp value
in position $b_r$, which represents
$\delta(\lcp(T_{\SA[r]},T_{\SA[r-1]})+1)$ and will be overwritten by
$\ell_b$. We use two auxiliary integer variables to store positions
$b_{r+1}$ and $e_{r+1}$. We compute $b_{r+1}$ and $e_{r+1}$ by scanning
$\LCP_2$ from $b_{s+1}$ up to finding $e_{r+1}$, by counting the encoding
lengths one by one. The values in $\LCP_2[b_{r+1}, e_{r+1}]$ are set to 0
and $\LCP_2[b_{s+1}, b_{r+1}-1]$ is shifted
$|\delta(\ell_a+1)|+|\delta(\ell_b+1)|-(e_{r+1}-b_{r+1})$ positions to the
left. To finish, the values of $\delta(\ell_a+1)$ and $\delta(\ell_b+1)$
are inserted into their corresponding positions $b_r$ and $b_{r+1}$ in
$\LCP_2$.

\begin{thm}
\label{thm:bwt_lcp2}
Given a string $T$ of length $n$, we can compute its $\BWT$ in-place and
$\LCP$ array compressed in $O(n \log \log n)$ bits, in the average case,
in $O(n^2)$ time using $O(1)$ additional space.
\end{thm}

\begin{proof}
The cost added by the modifications in Steps 2', 2'' and 4' is constant
since the encoding and decoding
operations are performed in $O(1)$ time and the left-shifting of the encoded
\lcp values in Step 4' is done word-size.
Therefore, the worst-case time complexity of the modified algorithm remains $O(n^2)$.
As for the space usage, the expected value of each \LCP array entry is $O(\log
n)$ for random texts~\cite{Fayolle2005} and for more specific domains, such as
genome sequences and natural language, this limit has been shown
empirically~\cite{Leonard2012}. 
Therefore, in the average case our \LCP array representation uses $O(n \log
\log n)$ bits, since we are using Elias $\delta$-coding~\cite{Elias1975}.
In the worst case, when the text has only the same symbols, the \LCP array
still requires $n \log n$ bits since $\sum_{i=0}^{n-1} \log(i) = \log (n!) =
\Theta(n \log n)$.  \hfill $\qedwhite$
\end{proof}

\section{Tradeoff}\label{s:tradeoff}

Crochemore \etal\ showed how, given \(k \leq n\), we can modify their algorithm
to run in \(O ((n^2 / k + n) \log k)\) time using \(O(k \sigma_k)\) space,
where $\sigma_k$ is the maximum number of distinct characters in a substring of
length $k$ of the text. 
The key idea is to insert characters from the text into the $\BWT$ in batches
of size $k$, thereby using \(O(1)\) scans over the $\BWT$ for each batch,
instead of for each character.  
Their algorithm can be modified further to output, for each batch of $k$
characters, a list of the $k$ positions where those characters should be
inserted into the current $\BWT$, and the position where the \$ should be
afterward~\cite{Juha}.  (This modification has not yet been implemented, so
neither has the tradeoff we describe below.)

From the list for a batch, with \(O(1)\) passes over the current $\BWT$ using
\(O(k \sigma_k)\) additional space, we can compute in \(O((n + k) \log k)\) time the
intervals in the current $\LCP$ array on which we should perform $\RMQ$s when
inserting that batch of characters and updating the $\LCP$ array, and with
\(O(1)\) more passes in \(O(n)\) time using \(O(k)\) additional space, we can perform
those $\RMQ$s.  
The only complication is that we may update the $\LCP$ array in the middle of
one of those intervals, possibly reducing the result of future $\RMQ$s on it.  
This is easy to handle with \(O(k)\) more additional space, however, and does not
change our bounds.  
Analogous to Crochemore \etal's tradeoff, therefore, we have the following
theorem:

\begin{thm}
\label{thm:tradeoff}
Given a string $T$ of length $n$ and \(k \leq n\), we can compute its $\BWT$ in-place 
and $\LCP$ array simultaneously in \(O((n^2 / k + n) \log k)\)
time using \(O(k \sigma_k)\) additional space, where $\sigma_k$ is the maximum number of
distinct characters in a substring of length $k$ of the text. 
\end{thm}

\section{Conclusion}\label{s:conclusion}

We have shown how to compute the \LCP array together with the \BWT using constant space.
Like its predecessor, our algorithm is quite simple and it builds on
interesting properties of the \BWT and of the \LCP array. 
Moreover, we show how to compute the \LCP array directly in compressed
representation with no asymptotic slowdown using Elias coding,
and we provide a time/space tradeoff for our algorithm when additional memory is
allowed.
We note that our algorithm can easily construct the suffix array using constant
space, with no overhead on the running time.
We also note that very recently there has been exciting work on obtaining
better bounds via randomization~\cite{Nicola}.

We leave as an open question whether our algorithm can be modified to compute
simultaneously the \BWT and the permuted \LCP in compressed form, which takes
only 2n + $o(n)$ bits, while using quadratic or better time and only $O(n)$ bits on
top of the space that initially holds the string and eventually holds the \BWT.

\section*{Acknowledgments}
We thank the anonymous reviewers and Meg Gagie for comments that improved the
presentation of the manuscript.  
We thank Djamal Belazzougui, Juha K\"arkk\"ainen, Gonzalo Navarro, Yakov
Nekrich and Nicola Prezza for helpful discussions.

Funding:
F.A.L. acknowledges the financial support of CAPES and CNPq [grant number 162338/2015-5].
T.G. acknowledges the support of the Academy of Finland and Basal Funds FB0001, Conicyt, Chile.
G.P.T. acknowledges the support of CNPq.


\end{document}